\documentclass[runningheads]{llncs}

\usepackage{graphicx} 
\usepackage{amssymb,amsmath,enumerate,nicefrac,color,fancyhdr,hyperref,adjustbox}
\usepackage{authblk}
\usepackage{algorithm,caption}
\usepackage{algpseudocode}
\usepackage[citestyle=numeric,bibstyle=numeric]{biblatex}
\usepackage{xspace}

\usepackage[margin=1in]{geometry}

\newtheorem{prop}[theorem]{Proposition}

\DeclareMathOperator*{\argmax}{arg\,max}

\renewcommand{\qed}{$\hfill\blacksquare$}

\DeclareMathOperator*{\size}{size}
\newcommand{\rank}{rank}

\newcommand{\var}{Var}
\newcommand{\rpm}{UPM\xspace}

\title{Sequential Linear Contracts on Matroids}
\author{Kanstantsin Pashkovich \and Jacob Skitsko \and Yun Xing}

 \institute{Department of Combinatorics and Optimization \\ University of Waterloo, Canada\\\email{kpashkovich@uwaterloo.ca} \quad \email{jskitsko@uwaterloo.ca} \quad \email{y39xing@uwaterloo.ca}
 }
 
\date{\today}

\begin{document}

\maketitle
\begin{abstract}
In this work, we study sequential contracts under matroid constraints. In the sequential setting, an agent can take actions one by one. After each action, the agent observes the stochastic value of the action and then decides which action to take next, if any. At the end, the agent decides what subset of taken actions to use for the principal's reward; and the principal receives the total value of this subset as a reward.  Taking each action induces a certain cost for the agent. Thus, to motivate the agent to take actions the principal is expected to offer an appropriate contract. A contract describes the payment from the principal to the agent as a function of the principal's reward obtained through the agent's actions. In this work, we concentrate on studying linear contracts, i.e.\ the contracts where the principal transfers a fraction of their total reward to the agent. We assume that the total principal's reward is calculated based on a subset of actions that forms an independent set in a given matroid. We establish a relationship between the problem of finding an optimal linear contract (or computing the corresponding principal's utility) and the so called matroid (un)reliability problem. Generally, the above problems turn out to be equivalent subject to adding parallel copies of elements to the given matroid.

\end{abstract}
\section{Introduction}
Consider the following  interaction between a principal and an agent. An investment firm, that we refer to as  a principal, is interested in managing their investment portfolios. To do the management of the investments, the firm  is hiring an agent, and  the principal offers a contract to the agent. According to the contract the principal commits to pay the agent an amount that depends on the quality of the agent's performance. In their turn, the agent must decide on a set of investments to make on behalf of the principal. 

Both the principal and the agent have the same stochastic information about the value of available investment opportunities. Before making any investment the agent is required to conduct an in depth survey for this investment to determine its true value, i.e.\  to determine the so called outcome of this investment. The agent may conduct the surveys for investments in any order they wish, but conducting each such survey is costly for the agent. 

Furthermore, there might exist some natural constraints on the set of investments to be made. For example, it is reasonable for the firm to require diversified investments so they ``do not put all of their eggs in one basket''. Thus, the investments made by the agent are potentially subject to some structural constraints. At any given point in time, after surveying some of the investment opportunities, the agent may decide that conducting further surveys is not worth the cost. The decision to stop conducting any further surveys is subject to such factors as: the costs of investment surveys, the information about the value of the surveyed investments, the stochastic information about not yet surveyed investments, and the structural constraints on the investments. Once the agent stops conducting surveys, the agent simply hands to the principal a best set of investments that the agent surveyed which satisfy the structural constraints. 

To make a best possible investment, the principal has to design an effective contract to motivate the agent in this setting. However, the stochastic nature of the available information and the private sequential decision making of the agent makes this contract design problem challenging for the principal. We consider this contract design problem, and show when the principal can obtain a good approximation to an optimal linear contract or to compute such a contract exactly.

\subsection{Sequential Model Description}\label{sec:model_description}

Let us describe the principal-agent model where the agent takes actions sequentially, and where the set returned by the agent is subject to feasibility constraints given by a matroid. 

\subsubsection*{Publicly Known Information.} 
We consider the principal-agent setting for matroid independent set problem.
Let $\mathcal M = (E, \mathcal I)$ be a matroid with a ground set $E$ and a collection of independent sets $\mathcal I$. Here, the total number of possible actions is $n$, i.e.\  $|E|=n$. For every element $i \in E$ its value is a non-negative random variable with support on at most $m$ possible \emph{outcome values}. One can probe the value of the element~$i \in E$ by paying a cost~$c_i \geq 0$.  For each element $i\in E$, the following information is known by both the agent and the principal:
\begin{enumerate}[(i)]
    \item the distribution $X_i\sim \mathcal F_i$, i.e.\  the distribution of the non-negative outcome value $X_i$ for the element $i$. The random variables $X_i$, $i\in E$ are independent.
    \item the cost $c_i$ of probing the element $i$, i.e.\  the cost $c_i$ that the agent needs to pay to probe the outcome value $X_i$. 
\end{enumerate}
We assume that each distribution $\mathcal F_i$, $i\in E$ is given by an explicit list of values $v_{i,k}$ and probabilities $p_{i,k}$, $k=1,\ldots,m$. Here, for each $i\in E$ and $k=1,\ldots,m$, the variable $X_i$ takes value $v_{i,k}$ with probability $p_{i,k}$, i.e.\  $\Pr[X_i=v_{i,k}]=p_{i,k}$. Note that if for some $i\in E$ the support of $\mathcal F_i$ has less than $m$ values, we add an appropriate amount of values $v_{i,k}$ with $p_{i,k}=0$.

We assume that we access $\mathcal M$ through an independence oracle, i.e.\  for every subset of $E$ we can get an answer whether it is an independent set with respect to $\mathcal M$.

\subsubsection*{Decisions of Principal and Agent.}
The principal asks the agent to find an independent set $I \in \mathcal{I}$ such that all elements in $I$ have their value probed by the agent.  The principal's  reward $X$ equals the total value of the found independent set, i.e.\ $X=\sum_{i\in I}X_i$. To encourage the agent, the principal proposes a contract $t: \mathbb{R}_{\geq 0 } \rightarrow \mathbb{R}_{\geq 0 }$. The contract maps the total value of an independent set $I$ to a non-negative payment to the agent. In other words, if the agent returns the principal an independent set of total value $X$, the principal pays the agent $t(X)$. In this paper, we consider a particular class of contracts, so called \emph{linear contracts}. A \emph{linear contract} $t_\alpha: \mathbb{R}_{\geq 0 } \rightarrow \mathbb{R}_{\geq 0 }$ is defined by a parameter $\alpha \in [0,1]$ such that  $t_\alpha(X) := \alpha \cdot X$ for every value of $X$. In other words, the principal gives an $\alpha$ fraction of the reward as a payment to the agent.

After the principal reveals the contract, the agent starts to probe elements sequentially, i.e.\ the agent probes elements one after another until the agent decides to stop probing. So, at each timepoint the agent can select an element to probe or decide to halt. For each probed element $i$ the agent has to pay the corresponding cost $c_i$.  When the agent halts, the agent has a set of probed items $S$, and then the agent selects an independent set $I \subseteq S$, $I \in \mathcal{I}$ to return to the principal. As outlined above, if the total value of $I$ equals $X$ then the agent receives the payment $t(X)$ from the principal.

The agent's strategy $\pi$ depends on the contract $t$, the matroid $\mathcal M=(E,\mathcal I)$, costs $c_e$, $e\in E$ and distributions $\mathcal F_i$, $i\in E$.
Given that, at each timepoint the strategy $\pi$ decides whether the agent has to halt or not, and if not then which element to probe next. Note, that at any given timepoint such decisions of $\pi$ depend on the set of elements probed so far and on their values. The agent is allowed to follow a randomized strategy.

\subsubsection*{Utility of Principal and Agent for Given Agent's Strategy.}
Consider an agent's strategy $\pi$. Let $S(\pi)$ be the random variable representing the set of probed elements. Let $I(\pi)$ be the random variable representing the independent set returned by the agent.  Recall, that  $I(\pi)$ is an independent set in $\mathcal M$, and $I(\pi)$ is always a subset of $S(\pi)$. Let $U_A(t,\pi)$ be the \emph{expected utility of the agent} under the contract $t$ and strategy $\pi$. Thus, for linear contracts  we have
\[ U_A(t_\alpha,\pi) := \mathbb{E}\left[\alpha\left(\sum_{i \in I(\pi)}X_i\right) - \sum_{i \in S(\pi)}c_i\right]\,.\]
Let $U_P(t,\pi)$ denote \emph{principal's expected utility} under  the  contract $t$ and agent's strategy $\pi$. For linear contracts, we have
\[ U_P(t_\alpha,\pi) := \mathbb{E}\left[\sum_{i \in I(\pi)}X_i - \alpha\left(\sum_{i \in I(\pi)}X_i\right)\right]\,.\]

\subsubsection*{Utility of Principal and Agent for Best Agent's Strategy.}
Given a contract $t$, the agent is interested in a strategy that maximizes the agent's expected utility. In other words, the agent is interested in a strategy $\pi^*(t)$ that is in the set $\argmax_{\pi} U_A(t,\pi)$. An agent's best response is a strategy $\pi^*(t)$ that maximizes the utility $U_P(t, \pi)$ subject to the constraint $\pi^*(t)\in \argmax_{\pi} U_A(t,\pi)$. In other words,  the agent selects a strategy that maximizes their utility while breaking the ties in favor of principal's utility, which is a standard assumption. 

We assume that given a contract $t$, the agent is always responding using a best response strategy $\pi^*(t)$. Thus, given a contract $t$ we can speak about the principal's utility $U_P(t)$  under this contract \[U_P(t) := U_P(t,\pi^*(t))\,.\]
Note that when the contract $t$ is clear from the context, we write $\pi^*$  for a best response instead of $\pi^*(t)$.

\subsection{Our results}

A central problem that we study in the current work is determining an optimal linear contract in a sequential model when the combinatorial constraints are given by a matroid. We call this problem the \emph{Optimal Linear Contract Problem on Matroids}. In all of the results below when we speak about $\size(\cdot)$ we refer to the size of the binary encoding of the corresponding number or the corresponding collection of numbers.

\begin{problem}[\textbf{OLCPM: Optimal Linear Contract Problem on Matroids}]\label{OLCPM}
  Let us be given a matroid $\mathcal{M} = (E,\mathcal{I})$, where every  element $i\in E$ has a cost $c_i$ and  an independent discrete distribution $\mathcal F_i$ with at most $m$ outcomes. For the principal-agent setting in a sequential model, 
 determine a linear  contract achieving the maximum principal's utility.
\end{problem}

In this work, we establish a connection between the above problem and the next problem defined on matroids, where each element is present with some probability. The well known Two Terminal Network Reliability Problem \cite{valiant1979complexity} is a special case of the next problem, when one considers graphic matroids. For more information on variants and generalizations of \emph{Matroid (Un)reliability Problem} see~\cite{brylawski_oxley_1992}, Section 6.3.E. In~\cite{brylawski_oxley_1992}, the authors formulate a more general problem, i.e.\ a problem of determining the probability that the existing elements form a set from a specified family and examine the relationship to the Tutte polynomial.

\begin{problem}[\textbf{\rpm: Unreliability Problem on Matroids}]\label{RPM}
    Let us be given a matroid $\mathcal{M} = (E,\mathcal{I})$ and an element $e\in E$.  Let every element $i\in E\setminus\{e\}$  have an independent probability $p_i$ of existing. Let $T$ be the random variable indicating the set of existing elements in $E\setminus\{e\}$. Determine the probability\[\Pr[\rank_{\mathcal M} (T)+1=\rank_{\mathcal M} (T\cup \{e\})]\,.\]
    In other words, determine the probability that $e$ is not spanned by the existing elements $T$.
\end{problem}

Let us state our first results. They essentially show that one can efficiently find an optimal linear contract, or an approximately optimal linear contract, for an instance of OLCPM as long as one is able to efficiently find an exact, or an approximate, solution for related  instances of \rpm.

\begin{theorem}\label{thm:OLCPM->RPM}
An input instance of OLCPM for a matroid $\mathcal{M}$ can be solved in polynomial time and with a polynomial number of calls to an oracle solving \rpm on the matroid $\mathcal M$, where each instance of \rpm has size polynomial in the input OLCPM instance size. 
\end{theorem}

\begin{theorem}\label{thm:OLCPM->RPM approx}
Let $\psi$ be a positive number. An input instance of OLCPM for a matroid $\mathcal{M}$ can be approximated up to the factor $(1+\psi)$  in polynomial time and with a polynomial number of calls to an oracle solving \rpm on the matroid $\mathcal M$ up to the approximation factor $(1+\psi)$, where each instance of \rpm has polynomial size in the input OLCPM instance size.
\end{theorem}

Now let us state the analogous results for reducing \rpm instances to OLCPM. 

\begin{theorem}\label{thm:OLCPM<-RPM}
An input instance of \rpm for a matroid $\mathcal{M}$ can be solved in polynomial time and with a polynomial number  of calls to an oracle solving OLCPM on some matroids $\mathcal M'$, where each instance of OLCPM has polynomial size in the input \rpm instance size. Here, each matroid $\mathcal M'$ is obtained from $\mathcal M$ by introducing parallel copies for some of the elements in $\mathcal M$.
\end{theorem}

\begin{theorem}\label{thm:OLCPM<-RPM approx}
Let $\psi$  be a positive number.
An input instance of \rpm can be approximated up to the factor $(1+\psi)^2$ in polynomial time and with a polynomial number of calls to an oracle solving OLCPM on some matroids $\mathcal M'$ up to the approximation factor $(1+\psi)$, where each instance of OLCPM has polynomial size in the input \rpm instance size. Here, each $\mathcal M'$ is obtained from $\mathcal M$ by introducing parallel copies for some of the elements in $\mathcal M$.
\end{theorem}

Given these results let us briefly mention some of the known results on \rpm and their implications for OLCPM.
First, we would like to note that the \rpm on uniform matroids can be efficiently solved in a straightfoward way. Indeed, the answer is the sum of the coefficients infront of $x^0$, $x^1$, \ldots, $x^{\rank(\mathcal M)-1}$ of the polynomial $\Pi_{e\in E}\left((1-p_e)+p_e\cdot x\right)$, so it can be efficiently computed by iterating on $e\in E$.

For laminar matroids, there is a polynomial time algorithm for solving  \rpm~\cite{dehaan}. As a result, we have a polynomial time algorithm for optimal linear contracts in laminar matroids. For graphic matroids it is \#P-hard to solve the \rpm even for planar graphs~\cite{1986scott}, so our results show that it is also  \#P-hard to solve the problem OLCPM.

In the paper~\cite{1996karger}, Karger studies an FPRAS for a problem closely related to \rpm on graphic matroids. Given a graphic matroid $\mathcal M$ and given probabilities for each edge $e\in E$ to exist, Karger provides an FPRAS for the probability that the resulting network is disconnected. Thus,~\cite{1996karger} provides an FPRAS for the so called Network (un)Reliability Problem.  To get an FPRAS for \rpm on graphic matroids we need a stronger result though: given two vertices we need an FPRAS to estimate the probability that the two vertices are not connected in the resulting network, which leads us to Two Terminal Network (un)Reliability Problem. For this problem, the results in~\cite{1996karger} provide an FPRAS only for pairs of vertices that are separated by a cut with a value close to the minimum cut value in a certain graph. Recent work has shown this two terminal problem is in general $\#BIS$-Hard \cite{feng2024fpras}. The class of $\#BIS$-Hard problems is made of problems equivalent to approximating the number of independent sets in a bipartite graph, and is a rich class of approximate counting/sampling problems \cite{dyer2004relative}. Resolving the approximation complexity of $\#BIS$ is a major open problem.

Finally, we would like to emphasize that OLCPM, see Problem~\ref{OLCPM}, is formulated as finding an optimal linear contract. From our proofs, it follows that analogues to Theorem~\ref{thm:OLCPM->RPM}, Theorem~\ref{thm:OLCPM->RPM approx}, Theorem~\ref{thm:OLCPM<-RPM}, Theorem~\ref{thm:OLCPM<-RPM approx} hold even when one is required to determine the maximum expected principal's utility achievable by a linear contract.

Now let us state our results that provide an FPRAS for two special cases of OLCPM.
\begin{theorem}\label{thm:FPRAS_special_cases}
    OLCPM admits an FPRAS in each of the following cases:
    \begin{enumerate}[(a)]
    \item \label{thm:FPRAS_special_cases balanced} there exists $\omega$ of polynomial size with respect to the input size such that  \[\frac{\max_{k=1,\ldots m}v_{i,k}p_{i,k}}{\max_{k=1,\ldots m}v_{j,k}p_{j,k}}\leq \omega\,\]
    holds for all distinct $i$, $j\in E$ such that $\mathcal F_j$ has in support not only $0$;
    \item \label{thm:FPRAS_special_cases support two} $m=2$ and for each $i\in E$ we have $v_{i,1}=0$; and additionally the cardinality  of the set \[\{v_{i,k}\,:\, i\in E,\, k=1,\ldots, m\}
    \] is bounded by a constant.
    \end{enumerate}
    \end{theorem}
    The case~\eqref{thm:FPRAS_special_cases balanced} of Theorem~\ref{thm:FPRAS_special_cases} covers such important cases as identical distributions, distributions where nonzero supports are balanced, etc. Also the case~\eqref{thm:FPRAS_special_cases support two} is one of the most natural cases, and at the first look one might expect to be able to obtain an FPRAS in this case. As we see later in Remark~\ref{rem:suport_three reduction}, the condition $m=2$ can not be extended to $m=3$ in the case~\eqref{thm:FPRAS_special_cases support two} without providing an FPRAS for the corresponding \rpm instances.

\subsection{Our Techniques}

Our results on reducing OLCPM to \rpm build on the ideas from~\cite{ezra2024contract} and~\cite{singla2018price}. Generally, in OLCPM, the FRUGAL algorithm from~\cite{singla2018price} allows us to compute the policy with which the agent would respond to a linear contract. In other words, the agent can use the FRUGAL algorithm from~\cite{singla2018price} to arrive at their optimal policy.  The main technical obstacle here is due to the tie breaking on the agent's side. Indeed, the agent has to select the policy that is most beneficial for the principal when there are several optimal agent's policies. To handle this issue we show that the required tie breaking forces the agent to select the policy with the largest cost whenever there are several optimal agent's policies. After establishing this, we handle tie breaking by introducing sufficiently small perturbations on the probing costs. After that we rely on the strategies analogous to the ones in~\cite{ezra2024contract} and~\cite{singla2018price}, to identify all critical values, and so to identify polynomially many linear contracts among which there is an optimal linear contract. 

To reduce OLCPM to \rpm, we carefully construct an instance of OLCPM that has only two candidates for an optimal linear contract. The construction is based on the interplay between the actions' costs and the actions' rewards. Roughly speaking, this interplay guarantees that if the principal wants to ``motivate" the agent to take one more action then the principal's utility from the previously taken actions is going to be enormously reduced by the increase in the transfer from the principal to the agent. Having such a construction, we guarantee that the principal either needs to ``motivate" the agent to take only one (the cheapest) action or has to ``motivate" the agent to potentially take actions until the very last action with a high expected value. This dichotomy allows us to reduce OLCPM to \rpm.

In the two special cases of OLCPM settings that we consider we provide an FPRAS. For this we need to show that for each linear contract we have an FPRAS for computing the expected principal's utility under this contract. In the case where all nonzero outcomes have ``roughly" the same expected value we show that a simple Monte Carlo algorithm leads to an FPRAS for computing the principal's utility under a given linear contract. In the case when each action has only one possible nonzero outcome and the total number of different outcomes over all actions is constant,  we show that if the simple Monte Carlo algorithm does not lead to an FPRAS then the considered linear contract is not optimal. This allows us to ``approximately" compare expected principal's utility under an optimal linear contract and under non-optimal linear contracts, and at the end to efficiently output an ``approximately" optimal linear contract. We also show that there are clear obstacles to extend this result to the case of two nonzero outcomes per action. For this, we modify our reduction from \rpm to OLCPM in such a way that the resulting OLCPM instances have only two nonzero outcomes per action and the total number of different outcomes over all actions is at most five.

\subsection{Related Work}
\subsubsection*{Sequential Contracts}
Our results build on previous work in the sequential contract setting \cite{ezra2024contract}. In \cite{ezra2024contract}, the authors introduced sequential contracts and considered the model where the agent returns to the principal the value of a single element, i.e.\ the value of a  single action, as a reward. Thus, the structural constraints in~\cite{ezra2024contract} correspond to the uniform matroid of rank one.  In~\cite{ezra2024contract}, the authors gave an algorithm for finding optimal linear contracts, and an algorithm for optimal general contracts when the number of outcomes is constant. Here, the number of outcomes refers to the number of potential values elements can take, i.e.\ the number of potential values of agent's actions. They complemented these positive results with hardness results for the case when the values of elements, i.e.\ values of actions, are allowed to be correlated. In particular, they showed the $\mathcal{NP}$-hardness of approximating an optimal linear contract and an optimal general contract  within any constant. Another line of work on sequential contracts appears in~\cite{hoefer2024contractdesignpandorasbox}. In their setting, the outcome of each action corresponds to a reward for the agent and a reward for the principal. The principal observes the reward and sees what action led to this award, but does not see what other actions the agent took. The principal  designs a contract that indicates how much money the principal transfers to the agent as a function of the received reward and the corresponding action. In~\cite{hoefer2024contractdesignpandorasbox}, the authors provide an exact algorithm for finding an optimal linear contract, and additionally provide an algorithm for finding an optimal general contract in certain special cases.

\subsubsection*{Pandora's Box and the Price of Information}
As observed in~\cite{ezra2024contract} and~\cite{hoefer2024contractdesignpandorasbox} that initiated the study of sequential contracts, Pandora's box problems~\cite{weitzman1978optimal} are deeply connected to the optimal agent's policies subject to a sequential contract. Indeed, \cite{weitzman1978optimal} described a costly search problem where the agent can take a sequence of actions while paying the cost for each of these actions and keeping the largest action value. Due to this definition, the Pandora's box problem introduced in~\cite{weitzman1978optimal}  described exactly the problem of finding the optimal agent's policy in \cite{ezra2024contract}. Pandora's box problems have been studied in further more general settings, e.g. \cite{singla2018price} considered the case when the agent can keep the values of several actions as long as the set of these actions satisfies some combinatorial constraints. \cite{singla2018price} gave an algorithm to solve the Pandora's box problem for the case when the combinatorial constraint corresponds to independence with respect to a matroid. This later setting corresponds to the problem of finding an optimal agent's policy in our setting. However, computing the agent's optimal reward turns out to be hard, to show this we establish connections to (un)reliability problems~\cite{provan1986complexity,valiant1979complexity,feng2024fpras,cen2024beyond}.

\subsubsection*{ Combinatorial Contract Theory}
Combinatorial structures naturally arise in contract design theory whenever an agent can take several actions, or there are several agents. Our setting is related to recent works in non-sequential combinatorial contract theory. \cite{babaioff2006combinatorial} considered combinatorial functions in the multi-agent setting and gave an algorithm to compute the optimal contract for AND networks. It was then observed by \cite{emek2012computing} that computing the optimal contract for OR networks was $\mathcal{NP}$-hard. There has been a recent surge of interest in contract theory problems, single agent combinatorial reward functions \cite{dutting2022combinatorial,dutting2024combinatorial} and multi-agent combinatorial reward functions\cite{dutting2023multi,deo2024supermodular,duetting2024multi,castiglioni2023multi,pashkovich2024linearcontractssupermodularfunctions}. Many of these results considered obtaining approximately optimal contracts, since obtaining optimal contracts turned out to be hard in many settings. The hardness of finding and approximating optimal contracts led to a series of interesting results~\cite{ezra2024approximability,dutting2024querycomplexitycontracts}.

An expansive number of different settings for contract theory have been explored in recent years, as contracts recently caught the eye of many researchers in the theoretical computer science community. As a small list, other recent models included contracts with inspections \cite{ezra2024contracts}, ambiguous contracts \cite{dutting2023ambiguous}, hiring for an uncertain task \cite{castiglioni2024hiringuncertaintaskjoint}, and agents choosing from menus of contracts \cite{guruganesh2034contracts, bernasconi2024agent}.

\subsection*{Acknowledgment} The research was supported in part by Discovery Grants RGPIN-2020-04346 and DGECR-2020-00265 from the Natural Sciences and Engineering Research
Council (NSERC) of Canada.

\section{Grade, Surrogate and FRUGAL Algorithm}\label{sec:grade_surrgoate}
Given  a linear contract $t_\alpha$, for each action $i\in E$ we define the value  \emph{grade} $\tau_i(\alpha)$ and the random variable \emph{surrogate} $Y_i(\alpha)$
following the terminology in~\cite{singla2018price}. Note that both $\tau_i(\alpha)$ and $Y_i(\alpha)$ depend on the linear contract~$t_\alpha$. 
\begin{definition}
Given a linear contract $t_\alpha$ and a random variable $X_i\sim \mathcal F_i$, we define the grade $\tau_i(\alpha)$ to be the unique solution to the equation \[\mathbb{E}[(\alpha X_i - \tau_i(\alpha))^+] = c_i\,,\] 
where $(X)^+$ denotes $\max(X,0)$.
Note if $c_i = 0$ the solution to the equation is not unique, and in this case we define $\tau_i(\alpha) := +\infty$. 
\end{definition}
Note that according to the above definition the value $\tau_i(\alpha)$  does not have to be nonnegative, and indeed $\tau_i(\alpha)$ is allowed to take negative values.

\begin{definition}
Given a linear contract $t_\alpha$ and a random variable $X_i\sim \mathcal F_i$, we define the surrogate $Y_i(\alpha) := \min\{\alpha X_i,\tau_i(\alpha)\}$. 
\end{definition}
Note, that for each action $i\in E$ the surrogate value $Y_i(\alpha)$ is revealed only when the value $X_i$ is revealed, i.e.\ only when $i$ is probed. When the value of $\alpha$ defining the contract $t_\alpha$ is clear from the context, we refer to $\tau_i(\alpha)$ and $Y_i(\alpha)$ as $\tau_i$ and $Y_i$.

\begin{algorithm} 
\caption{Frugal Algorithm, FRUGAL }
\label{ALG:FRUGAL}
\begin{algorithmic}[1] 
\State Start with independent set $I\leftarrow\varnothing$, probed set $S\leftarrow\varnothing$ and  $v_i\leftarrow 0$ for all elements $i\in E$.
\While{$\{i\in E\,:\,i\not \in I,\, \{i\}\cup I\, \in \mathcal I,\, v_i\geq  0  \}\neq \varnothing$}
\State For each element $i\not\in I$
\[v_i\leftarrow\begin{cases}
    Y_i(\alpha)&\text{if } i \in S\\
    \tau_i(\alpha)&\text{if }i\not \in S
\end{cases}\]
    \State Consider an element $j \in \argmax_i\{v_i\,:\, i\not \in I,\, \{i\}\cup I\, \in \mathcal I,\, v_i \geq 0 \}$ 
    \If{$j\in S$}
    \State Set $I\leftarrow I\cup \{j\}$, $v_j\leftarrow 0$
        \Else{} Probe $j$ and set $S\leftarrow S\cup\{j\}$
            \If{$X_j \geq \tau_j(\alpha)$}
                \State Set $I\leftarrow I\cup \{j\}$ and $v_j\leftarrow 0$.
            \EndIf
    \EndIf
\EndWhile
\end{algorithmic}
\end{algorithm}

The results from~\cite{singla2018price} imply that given a linear contract $t_\alpha$  the above FRUGAL algorithm, see Algorithm~\ref{ALG:FRUGAL}, maximizes the expected agent's utility. In other words, the algorithm FRUGAL probes the elements $S$ and returns the independent set $I$ such that the expected agent's utility is maximized.

\begin{prop}\label{prop:max_agent_utility}
    Given a linear contract $t_\alpha$, Algorithm~\ref{ALG:FRUGAL} maximizes the expected agent's utility.
    Moreover, the maximum expected agent's utility is equal to $\mathbb{E}[\max_{J\in \mathcal I} \sum_{j\in J}Y_j]$, i.e.\  the maximum expected agent's utility is equal to the expected maximum weight of an independent set in $\mathcal M$ subject to the weights $Y_j$, $j\in E$. 
\end{prop}

\section{Agent's Best Response Characterization}

In this section, we characterize the agent's best responses and show how to compute an agent's best response using the FRUGAL Algorithm, see Algorithm~\ref{ALG:FRUGAL}. In particular, we show that after a sufficiently small perturbation of the costs, the FRUGAL algorithm computes an agent's best response.

Given a contract $t_\alpha$, Proposition~\ref{prop:max_agent_utility} states that Algorithm~\ref{ALG:FRUGAL} computes an agent's policy that maximizes the agent's utility $U_A(t_\alpha,\pi)$. Recall, that an agent's best response is a policy that maximizes the agent's utility $U_A(t_\alpha,\pi)$ and then subject to that maximizes the principal's utility $U_P(t_\alpha,\pi)$.
The next lemma shows that when the agent is selecting a policy that maximizes the principal's utility $U_P(t_\alpha,\pi)$ among the policies that maximize the agent's utility $U_A(t_\alpha,\pi)$, the agent should choose a strategy that maximizes the expected cost $\sum_{i\in S(\pi)} c_i$.  Recall that given an agent's strategy $\pi$,  $S(\pi)$ is the random variable representing the set of probed elements and $I(\pi)$ is the random variable representing the returned independent set.

\begin{lemma}\label{lem:best_response_max_cost}
    Given a linear contract $t_\alpha$, let $\pi^*$ be a strategy that maximizes the agent's utility, i.e.\ $U_A(t_\alpha,\pi^*)=\max_{\pi} U_A(t_\alpha,\pi)$. Then $\pi^*$ is an agent's best response if and only if $\pi^*$ maximizes $\mathbb{E}\left[\sum_{i \in S(\pi)}c_i\right]$ among the strategies in  $\argmax_{\pi} U_A(t_\alpha,\pi)$.
\end{lemma}
\begin{proof}
Let $\beta$ be the maximum agent's utility under the contract $t_\alpha$, i.e.\ $\beta=\max_{\pi} U_A(t_\alpha,\pi)$. Let $\pi$ be an agent's strategy that maximizes the agent's utility, i.e.\ $U_A(t_\alpha,\pi)=\beta$. By definition of the agent's utility we have
\[ \beta = \mathbb{E}\left[ \alpha \sum_{i \in I(\pi)}X_i - \sum_{i \in S(\pi)}c_i \right] \,.\]
By definition of the principal's utility we have
\[
    U_p(t_\alpha,\pi)
    = (1-\alpha) \cdot \mathbb{E}\left[ \sum_{i \in I(\pi)}X_i \right]
    = \dfrac{1-\alpha}{\alpha} \cdot \left(\beta + \mathbb{E}\left[ \sum_{i \in S(\pi)}c_i \right] \right) \,.
\]
Since for a given contract $t_\alpha$, the values of  $\alpha$ and $\beta$ are fixed, i.e.\ do not depend on the agent's strategy, the maximization of  $U_p(t_\alpha,\pi)$ is equivalent to the maximization of  $\mathbb{E}[ \sum_{i \in S(\pi)}c_i ]$.
\qed\end{proof}

By Lemma~\ref{lem:best_response_max_cost}, an agent's best response is a strategy that first maximizes the agent's utility and subject to that maximizes the expected total cost of probing. We now show that with a perturbation on costs, the FRUGAL algorithm computes an agent's best response.

\begin{theorem}\label{thm:perturbed_cost_best_response}
Let us be given a linear contract $t_\alpha$. Then there exists $\varepsilon > 0$ such that for the modified agent's costs $c_i' := c_i(1-\varepsilon)$, $i\in E$ and with the same distributions $\mathcal F_i$, $i\in E$ there is a strategy $\pi'$ that maximizes the agent's utility.  Moreover, the strategy $\pi'$ corresponds to the strategy computed by the FRUGAL algorithm subject to the costs $c'_i$, $i\in E$, and the strategy $\pi'$ is an agent's best response subject to the original costs $c_i$, $i\in E$. Also, for any contract $t_\alpha$ such $\varepsilon$ can be computed efficiently, in particular it can be selected as the minimum nonzero value among $p_{i,k}$, $i\in E$, $k=1,\ldots,m$.
\end{theorem}
\begin{proof}Let us describe the way to obtain the desired $\varepsilon$, but before that we need to obtain some technical statements. First, note that for each $i\in E$ with $c_i\neq 0$  we have \[ \Pr[\alpha X_i - \tau_i \geq c_i] \geq \min_{k=1,\ldots,m\,:\, p_{i,k}\neq 0} p_{i,k}:=p_i\,.\] Indeed, the above inequality follows from $\mathbb{E}[(\alpha X_i - \tau_i )^+]=c_i$ in the definition of $\tau_i$.

Second, we want to show that, for every $i\in E$ with $c_i\neq 0$ and every $\varepsilon$, $0\leq \varepsilon<p_i$ the grade value $\tau'_i$ with respect to the costs $c_i'$ is sufficiently close to the grade value $\tau_i$ with respect to the cost $c_i$. In particular, let us show 
\begin{equation}\label{eq:perturbed_cost_change_in_grade}
\tau_i \leq \tau_i' \leq \tau_i + p_i^{-1}\varepsilon c_i\,.
\end{equation}
The inequality $\tau_i \leq \tau_i'$ follows from the definition of grade values and the fact that $c'_i\leq c_i$. To show the inequality $\tau_i' \leq \tau_i + p_i^{-1}\varepsilon c_i$ let us show that $\mathbb{E}[\left(\alpha X_i - (\tau_i+p_i^{-1}\varepsilon c_i)\right)^+] $ is smaller or equal to $c_i'$.
 Indeed, for every $i\in E$ and for every $\varepsilon$, $0\leq \varepsilon\leq p_i$, we have
\begin{align*}
    \qquad \mathbb{E}[\left(\alpha X_i - (\tau_i+p_i^{-1}\varepsilon c_i)\right)^+] 
    &= \sum_{k=1, \ldots,m} p_{i,k}\cdot(\alpha v_{i,k} - \tau_i-p_i^{-1}\varepsilon c_i)^+\\
    &\leq \sum_{k=1, \ldots,m} p_{i,k}\cdot(\alpha v_{i,k} - \tau_i)^+-\sum_{\substack{k=1, \ldots,m:\\ \alpha v_{i,k}\geq  \tau_i+p_i^{-1}\varepsilon c_i}}p_{i,k}\cdot p_i^{-1}\varepsilon c_i\\
    &\leq c_i-p_i\cdot p_i^{-1}\varepsilon c_i = c_i-\varepsilon c_i=c'_i\,,
\end{align*}
where the last inequality follows from $\Pr[\alpha X_i - \tau_i \geq c_i] \geq p_i$ and $p_i^{-1}\varepsilon c_i\leq c_i$. Note that to obtain the inequality~\eqref{eq:perturbed_cost_change_in_grade} it is sufficient to select $\varepsilon$ such that 
\begin{equation}\label{eq:perturbed_cost_epsilon1}
0<\varepsilon\leq\min_{i\in E}p_i\,.
\end{equation}

Let us now select $\varepsilon$ such that 
\begin{equation}\label{eq:perturbed_cost_epsilon2}
\max_{i\in E} \,p_i ^{-1} \varepsilon c_i \,< \,\min \{\mu\,:\, \mu =\alpha v_{j,k}-\tau_i \text{ for some }i,\, j\in E,\, k=1\ldots,m,\, \mu > 0  \}
\end{equation}
and
\begin{equation}\label{eq:perturbed_cost_epsilon3}
\max_{i\in E} \,p_i ^{-1} \varepsilon c_i \,< \,\min \{\mu\,:\, \mu =\tau_j-\tau_i \text{ for some }i,\, j\in E, \mu > 0  \}\,.
\end{equation}

At this point, let us select $\varepsilon$ such that~\eqref{eq:perturbed_cost_epsilon1}, \eqref{eq:perturbed_cost_epsilon2} and \eqref{eq:perturbed_cost_epsilon3} hold. Then, from~\eqref{eq:perturbed_cost_change_in_grade} and~\eqref{eq:perturbed_cost_epsilon2}  it is straightforward to see that for every realization and for every $i$, $j\in E$ if $Y_i>\tau_j$ then  $Y_i>\tau'_j$. Similarly, from~\eqref{eq:perturbed_cost_change_in_grade} and~\eqref{eq:perturbed_cost_epsilon3}  it is straightforward to see that for every realization and for every $i$, $j\in E$ if $\tau_i>\tau_j$ then  $\tau'_i>\tau'_j$. Thus, in this case every policy $\pi'$ computed by a FRUGAL algorithm subject to the costs $c'_i$, $i\in E$ is also a policy computed by a FRUGAL algorithm subject to the costs $c_i$, $i\in E$ and an appropriate tie breaking.

Let us compute the agent's utility $U'_A(t_\alpha,\pi)$ of a policy $\pi$ subject to the costs $c'_i$, $i\in E$
\begin{align*}
U'_A(t_\alpha,\pi)=\mathbb{E}[\sum_{i\in I(\pi)}\alpha X_i-\sum_{i\in S(\pi)} c'_i]=\mathbb{E}[\sum_{i\in I(\pi)}\alpha X_i-\sum_{i\in S(\pi)}(1-\varepsilon) c_i]=\\
\mathbb{E}[\sum_{i\in I(\pi)}\alpha X_i-\sum_{i\in S(\pi)} c_i]+\varepsilon \mathbb{E}[\sum_{i\in S(\pi)} c_i]=U_A(t_\alpha,\pi)+\varepsilon \mathbb{E}[\sum_{i\in S(\pi)} c_i].
\end{align*}
Note that from the above discussion and from Proposition~\ref{prop:max_agent_utility}, the policy $\pi'$ computed by the FRUGAL algorithm is a policy $\pi$ that maximizes both $U_A(t_\alpha,\pi)$ and $U'_A(t_\alpha,\pi)$. Then from the above equation it follows that $\pi'$ is a policy $\pi$ that maximizes $U_A(t_\alpha,\pi)$ and subject to that maximizes $ \mathbb{E}[\sum_{i\in S(\pi)} c_i]$. Thus the policy $\pi'$ is indeed a best agent's response.

\qed\end{proof}

\begin{remark}\label{rem:frugal_for_perturbed_costs}
Theorem~\ref{thm:perturbed_cost_best_response} and Proposition~\ref{prop:max_agent_utility} give us a way to compute a best agent's response in an efficient way. For this, we modify the costs as in Theorem~\ref{thm:perturbed_cost_best_response}  and then use the FRUGAL algorithm.
\end{remark}

\section{Critical Values}

In Section~\ref{sec:grade_surrgoate}, given a linear contract $t(\alpha)$ we defined the grade value $\tau_i$ and the surrogate random variable  $Y_i$ for each $i\in E$. Note, that both $\tau_i$ and  $Y_i$, $i\in E$ depend on the linear contract $t_\alpha$. To study how the grades and surrogates change with $\alpha$, let us denote them as $\tau_i(\alpha)$, $Y_i(\alpha)$, $i\in E$. Note that by definition, for every $i\in E$ we have $Y_i(\alpha)=\min\{\alpha X_i,\tau_i(\alpha)\}$, so it suffices to study the behaviour of the functions $\tau_i(\alpha)$, $i\in E$.

\begin{lemma}\label{lem:grade_piecewise_linear} For every $i\in E$, $\tau_i(\alpha)$ is a continuous piecewise linear function with at most $(m+1)$ segments.
\end{lemma}
\begin{proof}
For every $\alpha$ and $i\in E$ with $c_i\neq 0$, by definition of $\tau_i(\alpha)$ we have the following equation
\[\mathbb{E}[\left(\alpha X_i - \tau_i(\alpha)\right)^+] = \sum_{k =1, \ldots,m} p_{i,k}\cdot \left(\alpha v_{i,k}-\tau_i(\alpha)\right)^+= c_i\,.\]
Without loss of generality, we may assume $0 \leq v_{i,1} < \hdots< v_{i,m}$. 
\[
    \sum_{k \in [m]} p_{i,k} \max(v_{i,k}\alpha  - \tau_i(\alpha), 0) = c_i\,.
\]
To determine the value of $\tau_i(\alpha)$, we can first determine the largest $j$, $j=1,\ldots,m$ such that 
\[
    \sum_{k=j, \ldots, m} p_{i,k}(\alpha\cdot v_{i,k}  - \alpha\cdot  v_{i,j}) > c_i \,,
\]
and then we have
\[
\tau_i(\alpha)= \frac{\sum_{k=j,\ldots,m}p_{i,k}(\alpha\cdot v_{i,k}-c_i)}{\sum_{k=j,\ldots,m}p_{i,k}}\,.
\]
If no such $j$ exists then $\tau_i(\alpha)=-c_i+\sum_{k=j,\ldots,m}p_{i,k}\cdot\alpha\cdot v_{i,k}$.
This shows that the function $\tau(\alpha)$ is a piecewise linear function with at most $m+1$ segments. One can also observe the above computation for the values of the function $\tau(\alpha)$ show that it is a continuous function.
\qed\end{proof}

The next lemma shows that given costs $c_i$, $i\in E$ and distributions $X_i\sim \mathcal F_i$, $i\in E$, the domain of~$\alpha$ can be partitioned in $O(n^2m^2)$ intervals such that on each interval the outcome of the comparison between the values $\tau_i(\alpha)$, $i\in E$, the values $\alpha\cdot v_{i,k}$, $i\in E$, $k=1.\ldots,m$ does not depend on $\alpha$.\begin{lemma}\label{lem:interval_partition}
Let us consider the functions $v_{i,k}(\alpha):=v_{i,k}\cdot \alpha$ for $i\in E$ and $k=1.\ldots,m$ and also functions $\tau_i(\alpha)$ for $i\in E$. There are values $\gamma_1$, $\gamma_2$, \ldots, $\gamma_t$, $0\leq \gamma_1<\gamma_2< \ldots, < \gamma_t\leq 1$ such that $t\leq 20n^2m^2$ for every $j=0, \ldots, t$, every $\alpha'$, $\alpha'' \in [\gamma_j, \gamma_{j+1}]$ and for every two functions
\[
f'(\cdot), f''(\cdot)\,\in\, \left\{ v_{i,k}(\cdot)\,:\, i\in E,\, k=1,\ldots,m \right\}\cup \left\{ \tau_i(\cdot)\,:\, i\in E \right\}
\]
we have
\[
\left(f'(\alpha')-f''(\alpha')\right)\left(f'(\alpha'')-f''(\alpha'')\right)\geq 0\,.\]
Here, we define $\gamma_{0}:=0$ and $\gamma_{t+1}:=1$.

Moreover, such $\gamma_1$, $\gamma_2$, \ldots, $\gamma_t$ can be computed efficiently.
\end{lemma}\label{lem:critical_values}
\begin{proof}
    According to Lemma~\ref{lem:grade_piecewise_linear}, the functions in $\{\tau_i(\cdot)\,:\,i \in E\}$ are piecewise linear functions, with at most $(m+1)$ segments. Thus any two of them intersect at most $(m+1)^2$ times. Thus in total, there are at most $ (m+1)^2\cdot \binom{n}{2}$ intersection points of the functions in $\{\tau_i(\cdot)\,:\,i \in E\}$. The functions in $\left\{  v_{i,k}(\cdot)\,:\, i\in E,\, k=1,\ldots,m \right\}$ are linear functions. So the origin is the unique intersection point for any of the functions in $\left\{  v_{i,k}(\cdot)\,:\, i\in E,\, k=1,\ldots,m \right\}$,  there are at  most $m+1$ intersection points between a function in $\left\{ \tau_i(\cdot)\,:\, i\in E \right\}$ and a  function in $\left\{  v_{i,k}(\cdot)\,:\, i\in E,\, k=1,\ldots,m \right\}$. Thus in total we have at most $2n(n+1)(m+1)^2$, an do at most $20 n^2 m^2$ points of intersection.
    
Note, that all the intersection points above can be computed through an extensive inspection of pairs of functions in the lemma. To finish the proof, note that by Theorem~\ref{thm:perturbed_cost_best_response} given any pair of functions, we can efficiently compute all of their intersection points. 

\qed\end{proof}

We  apply Lemma~\ref{lem:interval_partition} for the perturbed costs $c'_i:=(1-\varepsilon)c_i$, $i\in E$ as in Remark~\ref{rem:frugal_for_perturbed_costs}. The values $\gamma_0=0$ and $\gamma_1$, $\gamma_2$, \ldots, $\gamma_t$, $\gamma_1<\gamma_2< \ldots < \gamma_t$ provided by Lemma~\ref{lem:interval_partition} are called \emph{critical values}. Note, that due to the design of the FRUGAL algorithm,  the impact of $\alpha$ on the probed set $S$ and on the returned set $I$ is limited to the comparison between the values $\tau_i$, $i\in E$ and random variables $Y_i$, $i\in E$. Lemma~\ref{lem:interval_partition} shows that for every $j=0,\ldots,t$ we can assume that for all $\alpha'$, $\alpha''$ in the interval $[\gamma_j, \gamma_{j+1}]$ and for each realization of $X_i$, $i\in E$, in the FRUGAL algorithm the probed set $S$ and the returned set $I$ are the same for both $\alpha'$ and $\alpha''$.  Now, since for each interval $[\gamma_j, \gamma_{j+1})$ and each realization of $X_i$, $i\in E$ the returned set is the same, to maximize their utility over the contracts $t_\alpha$, $\alpha\in [\gamma_j, \gamma_{j+1})$ the principal would select the contract $t_\alpha$ with $\alpha:=\gamma_j$. Indeed, this way the principal gets the same expected reward but minimizes the payment to the agent whenever the principal is restricted to select a contract $t_\alpha$, $\alpha\in [\gamma_j, \gamma_{j+1})$.

\section{Relationships between OLCPM and \rpm}

In this section, we show that the problems OLCPM and \rpm are essentially equivalent both in terms of finding the exact solution and in terms of approximating them.

\subsection{Reducing \rpm to OLCPM}\label{sec:RPM-to-OLCPM}
The next lemma allows us to efficiently  transform an instance of \rpm into another instance of \rpm, where all elements in the matroid have a sufficiently small probability of existing. This is done by replacing an element with large probability of existing by a collection of elements parallel to it, where each of these parallel elements has a small probability of existing. Our construction guarantees that the answer to the new instance of \rpm is at least the answer to the original instance of \rpm, but the difference between these answers is sufficiently small.

\begin{lemma}\label{lem:RPM_cleanup}
    Let us be given an instance of \rpm  for matroid $\mathcal{M} = (E,\mathcal{I})$ with $|E| = n$,  an element $e\in E$ and a probability $p_i$ of existing for each $i\in E\setminus\{e\}$. Let  $\delta$ and $\varepsilon$ be in $(0,1)$. There is an instance of \rpm given by a  matroid $\mathcal{M}' = (E',\mathcal{I}')$ with $|E'| = n'$,  an element $e'\in E'$ and with a probability $p'_i$ of existing for $i\in E'\setminus\{e'\}$. Moreover, the following statements are true:
    \begin{enumerate}[(i)]
        \item $\mathcal{M}'$ is obtained from $\mathcal M$ by replacing each element $i\in E\setminus\{e\}$ by several elements $P_i$, where each element in $P_i$  is parallel to the element $i$ and  where $|P_i|$ is at most $1$ if $p_i\leq \delta$ and  at most $\max(1,\lceil1/\log_{1-\delta}(1-p_i)\rceil)$ if $\delta<p_i<1$ and is at most $\max(1,\lceil1/\log_{1-\delta}(1-\varepsilon/n)\rceil)$ if $p_i=1$.
        
        \item For each $i\in E\setminus \{e\}$, we have $p_{i'}\leq \delta$ for all $i'\in P_i$.

        \item The outcome of the obtained \rpm instance is in $[\rho, \rho+\varepsilon]$ where $\rho$ is the  outcome of the original \rpm instance. Moreover, the obtained \rpm instance can be constructed in a polynomial time and  has size polynomial in $n$, $\size(p_{i}\,:\, i\in E\setminus\{e\})$, $\size(1/\delta)$, $\size(1/\varepsilon)$.
        
    \end{enumerate}
\end{lemma}

\begin{proof} 
First, let us update each $p_i$, $i\in E\setminus\{e\}$ to be $1-\varepsilon/n$ whenever $p_i=1$. Clearly such a change produces a new instance of \rpm where the outcome is in $[\rho, \rho+\varepsilon]$ due to the union bound.

The desired construction of $\mathcal{M}'$ can be obtained iteratively  as follows. Start with $E':=E$ and $p'_i:=p_i$ for each $i\in E\setminus \{e'\}$. Here, $e'$ is the special element $e$ in the original instance of \rpm.

For each $i\in E'\setminus\{e'\}$ with $\delta<p'_{i}<1$, we introduce two elements $j$, $j'$ that are parallel  to $i$ and update $E':=(E'\setminus \{i\})\cup \{j, j'\}$, $p'_j:=\delta$ and $p'_{j'}:=p'_{i}/(1-\delta)-\delta/(1-\delta)$. It is straightforward to check that the probability that at least one of the elements $j$, $j'$ exists equals $p'_j+(1-p'_j)p'_{j'}$, and so equals $p'_i$. Moreover, note that we have $(1-p'_{j'})/(1-p'_i)=1/(1-\delta)$; and so, every $i\in E\setminus \{e\}$ leads to  at most $\lceil1/\log_{1-\delta}(1-p_i)\rceil$ parallel elements in $\mathcal{M}'$ if $\delta<p_i<1$; and to at most $\lceil1/\log_{1-\delta}(1-\varepsilon/n)\rceil$ parallel elements in $\mathcal{M}'$ if $p_i=1$.

\qed\end{proof}

Given an instance of \rpm with sufficiently small probabilities for elements to exist, the next lemma produces an OLCPM instance on the same matroid. One of the crucial parameters in the construction is the number $\beta$. Let $\rho$ be the answer to the instance of \rpm. Then the constructed instance of OLCPM has only two critical values of interest: the critical value $0$ with the principal's utility $\beta$ and the critical value $1-\xi$ with the principal's utility $\rho+\varepsilon$, where $\xi$ and $\varepsilon$ are also parameters in the construction. This fact allows us to reduce \rpm to OLCPM using binary search on potential values of $\rho$ while setting $\varepsilon$ to be sufficiently small.

\begin{lemma}\label{lem:RPM_reduction}
     Let us be given an instance of \rpm  on a matroid $\mathcal{M} = (E,\mathcal{I})$ with $|E| = n\geq 2$ and  an element $e\in E$. Let us be given a probability $p_i$ of existing for each $i\in E\setminus\{e\}$. Let $\beta$ be in $[0,1]$ and $\varepsilon$ be in $(0,\frac{1}{2})$ and $\delta$  be in $(0,1)$ such that $\delta\frac{1/\varepsilon^{n}-1}{\varepsilon^2(1/\varepsilon-1)}<\beta$. Let us assume that for every $i\in E\setminus\{e\}$ we have $0<p_i\leq \delta$.    Let $\xi$ be such that $0<\xi<\varepsilon^{n-2}$ and \[ \xi\left(1+\frac{1/\varepsilon^{n}-1}{\varepsilon(1/\varepsilon-1)}\right)\leq \varepsilon\,.\]

     Then in polynomial time  in $n$, $\size(p_{i}\,:\, i\in E\setminus\{e\})$, $\size(\beta)$, $\size(1/\varepsilon)$, $\size(1/\delta)$, $\size(\xi)$  we can construct an instance of OLCPM on $\mathcal M$ of polynomial size with the following properties:
     \begin{enumerate}[(i)]
     \item In the constructed instance of OLCPM, there are at most $n+1$ critical values and they are among the values below, 
     
     \begin{itemize}
     
     \item $\gamma_0:=0$ and $\gamma_{n}:=1$

        \item $\gamma_1:=1-\varepsilon$, $\gamma_2:=1-\varepsilon^2$, \ldots, $\gamma_{n-2}:=1-\varepsilon^{n-2}$,
        \item $\gamma_{n-1}:=1-\xi$.
        
        \end{itemize}
     
     \item 
      In the constructed instance of OLCPM, the principal's utility for all critical values but $\gamma_0=0$ and $\gamma_{n-1}=1-\xi$ are at most $2\varepsilon \beta$. The principal's utility for the critical values $\gamma_0=0$ and $\gamma_{n-1}=1-\xi$ equals $\beta$ and $\rho'$, respectively,
     where $\rho'$ is some number in $ [\rho, \rho+\varepsilon]$ and $\rho$ is the outcome of the original \rpm instance.
     
     \end{enumerate}
    
     \end{lemma}
\begin{proof}

    Let us design the following OLCPM on the matroid $\mathcal{M}$. We identify $E\setminus \{e\}$ with the indices in $[n-1]$ and we identify $e$ with the index~$n$. 

     For each $i\in E$ we define a cost $c_i$ and we define a distribution $\mathcal F_i$ that consists of two outcomes: outcome $v_i$ with probability $p_i$ and outcome $0$ otherwise. For $i\in E\setminus \{e\}$ the values $v_i$ and  $c_i$ are going to be defined through the construction, while the values $p_i$ are taken from the original \rpm instance.
    
    For the element $i=1$, let us define  the cost $c_1:=0$ and the value $v_1:=\beta/p_1$. For every element $i= 2,\ldots,n-1$, let us define $c_i:=(1/\varepsilon^{i-1}-1) p_i$ and $v_i:=1/\varepsilon^{i-1}$. 
    
     The functions $\tau_i(\cdot)$ are such that for $i=2,\ldots,n-1$ we have $\tau_i(1-\varepsilon^{i-1})=0$. So for each $i\in [n-1]$, we have $\tau_i(\gamma_{i-1})=0$ and the function $\tau_i(\cdot)$ is linear on $[\gamma_{i-1},1]$. Also we have $\tau_i(1)=1$ for $i=2,\ldots, n$ and we have  $\tau_1(1)=+\infty>1$. 
     
     For the element $i=n$ we define $p_n:=1$, $c_n:=(1-\xi)/\xi$ and $v_n:=1/\xi$. We have $\tau_n(\gamma_{n-1})=0$,  $\tau_n(1)=1$ and the function $\tau_n(\cdot)$ is linear on $[\gamma_{n-1},1]$.

     Note that for every $i\in [n-1]$ we have $\gamma_i v_i\geq 1$. Indeed, for $i=1$ we have 
 \[
\gamma_1 v_1=  (1-\varepsilon)\frac{\beta}{\delta}\geq (1-\varepsilon) \frac{1/\varepsilon^{n}-1}{\varepsilon^2(1/\varepsilon-1)}  \geq 1\,.
 \]
 For $i=2,\ldots,n-2$ we have
  \[
\gamma_i v_i=  (1-\varepsilon^i)\frac{1}{\varepsilon^i} =\frac{1}{\varepsilon^i}  -1 \geq 1\,.
 \]
 
Let us make the crucial observation for the proof.  For every $\alpha \in [0,1)$ the element $n$ is not probed if and only if $\tau_n(\alpha)<0$ or the set $\{i\in [n-1]\,:\, X_i=v_i\}$ contains the element $n$ in its span with respect to the matroid $\mathcal M$. This is due to the fact that for each  $i\in [n-1]$ and every $\alpha\in [0,1)$ if $\tau_n(\alpha)\geq 0$ then we have $\alpha v_i\geq 1=\tau_n(1)$.
     
    Now it is straightforward to check that  we have only $n+1$ critical threshold values, each belonging to one of three types below:
    \begin{enumerate}[(i)]
        \item $\gamma_0=0$ and $\gamma_{n}=1$

        \item $\gamma_1=1-\varepsilon$, $\gamma_2=1-\varepsilon^2$, \ldots, $\gamma_{n-2}=1-\varepsilon^{n-2}$
        \item $\gamma_{n-1}=1-\xi$.
    \end{enumerate} 
    
   Let $\nu_i$, $i=0,\ldots, n$ denote the principal's utility $U_P(t_{\gamma_i})$ corresponding to the value $\gamma_i$. To estimate the utility note that the element $i\in[n]$ can be probed only for critical values that are at least $\gamma_{i-1}$. 
   \begin{enumerate}[(i)]
        \item $\nu_0=\beta$ and $\nu_{n+1}=0$

        \item for $i=1,\ldots, n-2$, we have
        \begin{align*}
        \nu _i\leq &\varepsilon^{i}(v_1p_1+v_2p_2+v_3p_3+\ldots+v_{i+1}p_{i+1})\leq\\
        &\varepsilon^{i}(\beta+\delta/\varepsilon^{1}+\delta/\varepsilon^{2}+\ldots+ \delta/\varepsilon^{i})\leq\\ 
        &\varepsilon \beta+\delta\frac{1/\varepsilon^{i}-1}{\varepsilon(1/\varepsilon-1)}\leq 2\varepsilon\beta
        \end{align*}
        \item  for $i=n-1$, we have\[\nu_{n-1}\geq(1-\gamma_{n-1})v_n p_n\rho=\rho  \]
        and we also have 
        \begin{align*}
        \nu_{n-1}\leq (1-\gamma_{n-1})(v_1p_1+v_2p_2+v_3p_3+\ldots+v_{n-1}p_{n-1})+(1-\gamma_{n-1})v_np_n\rho \leq\\
        \xi\left(\beta+\frac{1/\varepsilon^{n}-1}{\varepsilon(1/\varepsilon-1)}\right)+ \rho\leq \varepsilon+\rho\,.  
        \end{align*}
        Thus, we have \[\rho\leq \nu_{n-1}\leq \rho+\varepsilon\,.\]
    \end{enumerate}

\qed\end{proof}

\begin{remark}
In the proofs of  Theorem~\ref{thm:OLCPM<-RPM} and Theorem~\ref{thm:OLCPM<-RPM approx}, we assume that oracles for solving OLCPM output a critical value. Note that these proofs are working also when the oracles output not a critical value. This is due to the fact that for two consecutive critical values $\gamma_j$, $\gamma_{j+1}$ we have that $U_P(t_\alpha)$ is linear function of $\alpha$ on $[\gamma_j, \gamma_{j+1})$. Also in Lemma~\ref{lem:RPM_reduction} for all critical values except $\gamma_0=0$ and $\gamma_{n-1}=1-\xi$ the principal's utility is at most $2\varepsilon\beta$. So from an OLCPM oracle that is not restricted to output a critical value, we can always obtain an OLCPM oracle that outputs a critical value by checking whether the output is in the interval $[\gamma_0, \gamma_1)$ or $(\gamma_{n-2}, \gamma_n)$.
\end{remark}

\begin{proof}[Proof of Theorem~\ref{thm:OLCPM<-RPM}]

Note that we can efficiently check whether the answer for an \rpm instance equals $0$. This is because the answer is $0$ if and only if the special element $e$ is spanned by the elements $\{i\in E\setminus\{e\}\,:\, p_i=1\}$ with respect to the matroids $\mathcal M$.

Let us be given an instance of \rpm  for a matroid $\mathcal{M} = (E,\mathcal{I})$ with $|E| = n\geq 2$ and  an element $e\in E$. Let us be given a probability $p_i=r_i/d_i$ of existing for each $i\in E\setminus\{e\}$, where both $r_i$ and $d_i$ are natural numbers.   Let $\rho$ be the answer to the given \rpm. 
Then we define $\Lambda:=\Pi_{i\in E\setminus\{e\}}d_i$ and
$\varepsilon:=\frac{1}{4\Lambda}$.

Then using oracle calls for solving OLCPM and  binary search on $\beta$ from \[\{z/\Lambda\,:\, z=1, \ldots, \Lambda-1, \Lambda\},\,\]
we can determine the value of $\rho$ due to Lemma~\ref{lem:RPM_cleanup} and Lemma~\ref{lem:RPM_reduction}. For Lemma~\ref{lem:RPM_cleanup} and Lemma~\ref{lem:RPM_reduction}, we select $\delta$ and $\xi$ to satisfy the conditions  $\delta\frac{1/\varepsilon^{n}-1}{\varepsilon^2(1/\varepsilon-1)}<\beta$ and  
$ \xi\left(1+\frac{1/\varepsilon^{n}-1}{\varepsilon(1/\varepsilon-1)}\right)\leq \varepsilon$
in the statement of  Lemma~\ref{lem:RPM_reduction}. 

Using binary search we can find two consecutive values $\beta'$, $\beta''$, $\beta'<\beta''$  in the set  \[\{z/\Lambda\,:\, z=1, \ldots, \Lambda-1, \Lambda\},\,\]
such that for $\beta'$ the optimal contract corresponds to $0$ and for $\beta''$ the optimal contract corresponds to $1-\xi$ for $\xi$ given by $\beta''$. Then $\beta''$ is the correct answer for the original \rpm instance.
    
\qed\end{proof}

\begin{proof}[Proof of Theorem~\ref{thm:OLCPM<-RPM approx}]

As in the proof of Theorem~\ref{thm:OLCPM<-RPM}, we can efficiently check wether the answer for the \rpm instance equals $0$.

Let us be given an instance of \rpm  for a matroid $\mathcal{M} = (E,\mathcal{I})$ with $|E| = n\geq 2$ and  an element $e\in E$. Let us be given a probability $p_i=r_i/d_i$ of existing for each $i\in E\setminus\{e\}$, where both $r_i$ and $d_i$ are natural numbers.   Let $\rho$ be the answer to the given \rpm instance. 

Let us show that using polynomial number of oracle calls that approximate  OLCPM up to the factor $(1+\psi)$ we can approximate $\rho$ up to the factor $(1+\psi)^2$. For this we again define $\Lambda:=\Pi_{i\in E\setminus\{e\}}d_i$, and let us further define $\varepsilon$ to be a positive number not larger than $\psi/\Lambda$, i.e. $1/\Lambda+\varepsilon\leq (1+\psi)/\Lambda$. 

Then using oracle calls for solving OLCPM and a binary search on $\beta$ from 
\[\{(1+\psi)^z/\Lambda\,:\, z=0, \ldots, \lceil \log_{1+\psi}\Lambda\rceil\},\,\]
we can determine the value of $\rho$ using Lemma~\ref{lem:RPM_cleanup} and Lemma~\ref{lem:RPM_reduction}. 
For Lemma~\ref{lem:RPM_cleanup} and Lemma~\ref{lem:RPM_reduction}, we select $\delta$ and $\xi$ to satisfy the conditions  $\delta\frac{1/\varepsilon^{n}-1}{\varepsilon^2(1/\varepsilon-1)}<\beta$ and  
$ \xi\left(1+\frac{1/\varepsilon^{n}-1}{\varepsilon(1/\varepsilon-1)}\right)\leq \varepsilon$
in the statement of  Lemma~\ref{lem:RPM_reduction}. 

Using binary search we can find two consecutive values $\beta'$, $\beta''$, $\beta'<\beta''$  in the set  \[\{(1+\psi)^z/\Lambda\,:\, z=0, \ldots, \lceil \log_{1+\psi}\Lambda\rceil\},\,\]
such that for $\beta''$ the $(1+\psi)$-approximately optimal contract output by the OLCPM oracle corresponds to $0$ and for $\beta'$ the $(1+\psi)$-approximately optimal contract corresponds to $1-\xi$ for $\xi$ given by $\beta'$.  Thus, we have $\beta'\leq (1+\psi)(\rho+\varepsilon)$ and $(1+\psi)\rho\leq \beta''$.
Note that by the definition of $\varepsilon$ we have $\rho+\varepsilon\leq (1+\psi)\rho$, since $\rho\geq 1/\Lambda$. Hence, 
\[(1+\psi)\rho\leq \beta''\quad\text{ and }\quad \beta'\leq (1+\psi)^2\rho,\,\]
and so both $\beta'$ and $\beta''$ provides an $(1+\psi)^2$-approximation for $\rho$. 

\qed\end{proof}

\subsection{Reducing  OLCPM to \rpm}
\label{sec:OLCPM->RPM}

To reduce OLCPM to \rpm and prove Theorem~\ref{thm:OLCPM->RPM} and Theorem~\ref{thm:OLCPM->RPM approx} it is enough to show how to  estimate $U_P(t_\alpha)$ for each critical value $\alpha$. For this, given $\alpha\in (0,1)$ we show how to estimate the probability that FRUGAL Algorithm accepts the element $i\in E$ with the value $v_{i,k}$ for all $i\in E$ and $k=1,\ldots,m$.

Since FRUGAL Algorithm uses comparison of different values we do a tie breaking between the surrogates and grades as follows. For the comparison, we associate $E$ with $\{1, 2,\ldots, n\}$. We will assume that we have a tie-breaking comparison that works in the following way, for $i$, $j\in E$, $i\neq j$ we have $\tau_i(\alpha) \succ
\tau_j(\alpha)$ if $\tau_i(\alpha) >
\tau_j(\alpha)$ or if $\tau_i(\alpha) =
\tau_j(\alpha)$ and $i<j$. Similarly, for $i$, $j\in E$, $i\neq j$ we have $Y_i(\alpha) \succ
Y_j(\alpha)$ if $Y_i(\alpha) >
Y_j(\alpha)$ or if $Y_i(\alpha) =
Y_j(\alpha)$ and $i<j$. Finally, for $i$, $j\in E$, $i\neq j$ we have $Y_i(\alpha) \succ
\tau_j(\alpha)$ if $Y_i(\alpha) \geq 
\tau_j(\alpha)$.  We define $\argmax^\succ$ as $\argmax$ except that the ties are broken according to $\succ$.

Let us be given a contract $t_\alpha$ and a number $\xi>0$. Let $r_{i,k}(\alpha)$, $i\in E$, $k=1,\ldots,m$ be the probability of the event that the policy produced by the FRUGAL algorithm (with tie-breaking based on $\succ$) includes $i$ in the output  conditioned on $X_i=v_{i,k}$, i.e.\ $r_{i,k}(\alpha)$ is the probability  that $i\in I$ at the end of the FRUGAL algorithm conditioned on $X_i=v_{i,k}$.

First, let us state when an element $i\in E$ with $X_i=v_{i,k}$, $k=1,\ldots,m$ is included in the output. The next remark is based on the structure of the FRUGAL Algorithm, i.e.\ how it probes elements and how it includes them in the output set $I$.

Recall, that when the linear contract $t_\alpha$ is clear from the context we write $r_{i,k}$, $\tau_i$, $Y_i$ for $i\in E$ and $k=1,\ldots,m$.

\begin{remark}\label{rem:element_spanned}
    Let us be given a contract $t_\alpha$. Let us be given an element $i\in E$ with $X_i=v_{i,k}$, $k=1,\ldots,m$. Then the element $i$ is included in the output of the FRUGAL Algorithm with tie-breaking defined by $\succ$ if and only if $\tau_i\geq 0$ and $i$ does not lie in the span of the following set
    \[
    \left\{j\,:\, \tau_j\succ\tau_i\,, Y_j\succ \tau_i,\, j\neq i\right\}\cup \left\{j\,:\, \tau_i\succ \tau_j\succ Y_i\,, Y_j\succ Y_i,\, j\neq i\right\}\,.
    \]
\end{remark}
\begin{proof}

Note, the element $i$ cannot be probed if $\tau_i<0$.
The set of elements \[\left\{j\,:\, \tau_j\succ\tau_i\,, Y_j\succ \tau_i,\, j\neq i\right\}\] corresponds to the elements that are included in $I$ before $i$ is probed. 
The set of elements \[\left\{j\,:\, \tau_i\succ \tau_j\succ Y_i\,, Y_j\succ Y_i,\, j\neq i\right\}\] corresponds to the elements that are included in $I$ after $i$ is probed and before $i$ is considered for being included in $I$. 
    
\qed\end{proof}

Using Remark~\ref{rem:element_spanned}, it is straightforward to reduce the computation of the probability that the element $i\in E$ is accepted by the FRUGAL Algorithm with the value $v_{i,k}$, $k=1,\ldots,m$ to an \rpm problem of polynomial size with respect to the size of the original OLCPM problem.

\section{Optimal Linear Contracts for Special Cases}
In this section, we continue the discussion from Section~\ref{sec:OLCPM->RPM} to prove Theorem~\ref{thm:FPRAS_special_cases}. For this we show how to estimate $U_P(t_\alpha)$ for every critical value $\alpha$ in these special cases of OLCPM.

\subsection{Balanced Outcome Values}

For both special cases we consider a straightforward Monte Carlo algorithm.  Lemma~\ref{lem:sampling} formalizes that when the answer in the instance of \rpm (corresponding to some instance of OLCPM and some $\alpha$) is not too small then it can be efficiently approximated using several independent runs of Algorithm~\ref{ALG:FRUGAL}. Here, in the independent runs of Algorithm~\ref{ALG:FRUGAL} we do tie-breaking defined by $\succ$. For the sake of completeness we include the proof.

\begin{lemma}\label{lem:sampling}
Let us be given a contract $t_\alpha$ and numbers $\mu > 0, \xi>0$. For each $i\in E$ let $\rho_{i,k}$, be the number computed by Algorithm~\ref{ALG:FRUGAL_SAMPLE}, which is a random variable due to sampling.

Then for each $i\in E$ such that $r_{i,k}\geq \mu$ we have $\Pr[|\rho_{i,k}-r_{i,k}|\geq \mu\cdot r_{i,k}]\leq \frac{1}{80m^3n^3}$.
\end{lemma}

\begin{proof}

    Algorithm~\ref{ALG:FRUGAL_SAMPLE} corresponds to running $80m^3n^3/\mu^4$ times Algorithm~\ref{ALG:FRUGAL} and using $R_{i,k}$ to count how many times the element $i$ is accepted conditioned on $X_i=v_{i,k}$ for each $i\in E$, $k=1,\ldots,m$. So, we have $r_{i,k}=\mathbb{E}[\rho_{i,k}]$. Also since $R_{i,k}$ is a sum of $80m^3n^3/\mu^4$ independent binary random variables we have \[\var(\rho_{i,k})=\left(\frac{\mu^4}{80m^3n^3}\right)^2 \var(R_{i,k})\leq \left(\frac{\mu^4}{80m^3n^3}\right)\,.\] 
    Thus by Chebyshev's inequality, if $r_{i,k}\geq \mu$ we have 
    \[\Pr[|\rho_{i,k}-r_{i,k}|\geq \mu\cdot r_{i,k}]\leq \frac{\mu^4}{80m^3n^3}\cdot\frac{1}{\mu^2 r_{i,k}^2}\leq \frac{\mu^4}{80m^3n^3}\cdot\frac{1}{\mu^4} \leq \frac{1}{80m^3n^3}\,,\]
    finishing the proof.
    
\qed\end{proof}

\begin{algorithm}[h] 
\caption{Sampling Algorithm}
\label{ALG:FRUGAL_SAMPLE}
\begin{algorithmic}[1] 
\State Set $t\gets 0$
\State For each $i\in E$ and $k=1,\ldots,m$ set $R_{i,k}\leftarrow 0$
\While{$t\leq 80m^3n^3/\mu^4$}
\State Set $t\leftarrow t+1$
\State For each $i\in E$ draw the value $X_i\sim \mathcal F_i$ independently, and set $Y_i \gets \min\{\alpha\cdot X_i,\tau_i\}$
\For{each $i\in E$ and each $k=1,\ldots,m$}
    \State $Y'_i \gets \min\{\alpha\cdot v_{i,k},\tau_i\}$
    \State $S \gets \left\{j\,:\, \tau_j\succ\tau_i\,, Y_j\succ \tau_i,\, j\neq i\right\}\cup \left\{j\,:\, \tau_i\succ\tau_j\succ Y'_i\,, Y_j\succ Y'_i,\, j\neq i\right\}$ 
    \If{$\tau_i\geq 0$ and $i$ is not spanned by $S$}
        \State $R_{i,k} \gets R_{i,k} + p_{i,k}$
    \EndIf
\EndFor
\EndWhile
\State For each $i\in E$ set $\rho_{i,k}\leftarrow R_{i,k}/t$
\end{algorithmic}
\end{algorithm}

Intuitively, the next lemma shows that the only obstacle in estimating the principal's utility is due to the elements that can have large expected value but small probability of being accepted with this value due to being ``blocked'' through other elements.

\begin{lemma}\label{lem:ignoring}
    Let us be given a contract $t_\alpha$ and numbers $\mu>0$, $\varepsilon>0$ with $2\mu<\varepsilon$. If for each $i\in I$ with $\tau_i\geq 0$ and each $k=1,\ldots,m$, we have $(1-\alpha)v_{i,k}p_{i,k}\leq \frac{1}{160m^4n^4} U_P(t_{\alpha}) \frac{\varepsilon}{\mu}$ or $r_{i,k}\geq \mu$ then with probability at least $1-\frac{1}{80 m^2n^2}$ we have
    \[
    (1-\varepsilon) U_P(t_\alpha)\leq(1-\alpha)\sum_{\substack{i\in E,\\k=1,\ldots,m }}v_{i,k} p_{i,k} \rho_{i,k}\leq (1+\varepsilon) U_P(t_\alpha)\,,
    \]
    where $\rho_{i,k}$ are the random variable computed by Algorithm~\ref{ALG:FRUGAL_SAMPLE}.
\end{lemma}
\begin{proof}
To prove the  lemma we use the same notation as in Lemma~\ref{lem:sampling}.
Note that $U_P(t_\alpha)$ equals \[(1-\alpha)\sum_{\substack{i\in E,\\k=1,\ldots,m }}v_{i,k} p_{i,k} r_{i,k}\,.\]
Thus we can compare each summand separately. For each $i\in I$ and $k=1,\ldots,m$ if $r_{i,k}\geq \mu$ then by Lemma~\ref{lem:sampling} we have 
\[\Pr[|\rho_{i,k}-r_{i,k}|\geq \frac{1}{2}\varepsilon\cdot r_{i,k}]\leq \frac{1}{80m^3n^3}\,.\]
So for each $i\in I$ and $k=1,\ldots,m$ if $r_{i,k}\geq \mu$ then with probability at least $1-\frac{1}{80m^3n^3}$ we have
\[ \left(1-\frac{1}{2}\varepsilon\right) (1-\alpha) v_{i,k}p_{i,k} r_{i,k}\leq (1-\alpha) v_{i,k}p_{i,k} \rho_{i,k}\leq \left(1+\frac{1}{2}\varepsilon\right) (1-\alpha) v_{i,k}p_{i,k} r_{i,k}\]

For each $i\in I$ and $k=1,\ldots,m$ if $r_{i,k}< \mu$ then by Markov's inequality we have 
\[
\Pr[\rho_{i,k}\geq 80 m^3n^3 \mu]\leq \frac{1}{80 m^3n^3}\,.
\]
Hence, for each $i\in I$ and $k=1,\ldots,m$ if $r_{i,k}< \mu$  then with probability at least $1-\frac{1}{80m^3n^3}$ we have
\[
0\leq (1-\alpha)v_{i,k}p_{i,k} \rho_{i,k}\leq (1-\alpha)\left(\frac{1}{160m^4n^4(1-\alpha)} U_P(t_{\alpha}) \frac{\varepsilon}{\mu}\right) \left(80 m^3n^3\mu\right) 
\leq  \frac{\varepsilon}{2mn}U_P(t_{\alpha})\,,\]
i.e.\ we have 
    \[
0\leq (1-\alpha)v_{i,k}p_{i.k} \rho_{i,k}\leq  \frac{\varepsilon}{2mn}U_P(t_{\alpha})\,.\] 
Note, that a similar estimation shows that if $r_{i,k}<\mu$ then $(1-\alpha)v_{i,k}p_{i,k} r_{i,k}$ is at most $\frac{\varepsilon}{2mn}U_P(t_{\alpha})$. So for each  $i\in I$ and $k=1,\ldots,m$ if $r_{i,k}< \mu$ then both $(1-\alpha)v_{i,k}p_{i,k} r_{i,k}$ and $(1-\alpha)v_{i,k}p_{i,k} \rho_{i,k}$ are at most $\frac{\varepsilon}{2mn}U_P(t_{\alpha})$ with probability at least $1-\frac{1}{80m^3n^3}$. Thus, with probability at least $1-\frac{1}{80m^3n^3}$ we have $(1-\alpha)v_{i,k}p_{i,k} |r_{i,k}-\rho_{i,k}|\leq \frac{\varepsilon}{2mn}U_P(t_{\alpha})$.

Now, the statement of the lemma follows by the union bound.

\qed\end{proof}

\begin{proof}[Proof of Theorem~\ref{thm:FPRAS_special_cases} Case~\eqref{thm:FPRAS_special_cases balanced}]
    Let us assume that there exists $\omega>0$ of polynomial size with respect to the size of OLCPM such that  \[\frac{\max_{k=1,\ldots m}v_{i,k}p_{i,k}}{\max_{k=1,\ldots m}v_{j,k}p_{j,k}}\leq \omega\,\]
    for all distinct $i$, $j\in E$ such that $\mathcal F_j$ has in support not only $0$. Let us define $\mu$ as follows
    \[\mu:=\frac{\varepsilon}{2\omega}\frac{1}{160m^4n^4}\]
    
    For every $\alpha$, $i\in E$ and $k=1,\ldots,m$ if $r_{i,k}<\mu$ and $\tau_i\geq 0$ then we have \[U_P(t_\alpha)\geq (1-\alpha)(1-\mu) \min_{j\in E}\left(\sum_{k=1}^m v_{j,k}p_{j,k}\right)\,.\] Indeed, this is due to the fact that $U_P(t_\alpha)$ is produced by the FRUGAL Algorithm attempting to output at least one element (other than the element  $i$) with probability at least $1-\mu$. Then we have \[(1-\alpha)v_{i,k}p_{i,k}\leq (1-\alpha)v_{i,k}\leq (1-\alpha)\omega\left(\min_{j\in E}\sum_{k=1}^m v_{j,k}p_{j,k}\right) \leq \frac{1}{160m^4n^4} U_P(t_{\alpha}) \frac{\varepsilon}{\mu}\,.\]
    So by Lemma~\ref{lem:ignoring} for each $\alpha$ with probability 
    at least $1-\frac{1}{80m^2n^2}$ we have    \begin{equation}\label{eq:sample_correct_critical_value}
    (1-\varepsilon) U_P(t_\alpha)\leq(1-\alpha)\sum_{\substack{i\in E,\\k=1,\ldots,m }}v_{i,k} p_{i,k} \rho_{i,k}\leq (1+\varepsilon) U_P(t_\alpha)\,.
    \end{equation} In total, we have at most $10n^2m^2+1$ critical values by Lemma~\ref{lem:critical_values} that can be computed efficiently. Thus, by union bound  with probability at least $3/4$  we have that~\eqref{eq:sample_correct_critical_value} holds for each critical value $\alpha$.
    
\qed\end{proof}

\subsection{Constant Number of Outcome Values}

The next lemma studies a linear contract such that there is an element which contributes a large value $U$ to the utility of the principal (conditioned on the element being probed), but the element's probability to be  probed is small. In this case, the lemma shows that another linear contract achieves a  principal's utility comparable to $U$ (up to loss of a certain factor).
\begin{lemma}\label{lem:change_critical_value}
    Let us be given a contract $t_\alpha$ and numbers $\varepsilon>0$ and $U$. For each $i\in E$ the distribution $\mathcal F_i$ corresponds to taking value $0$ with probability $1-p_i$ and value $v_i$ with probability $p_i$; and  the cardinality of $\{v_i\,:\, i\in E\}$ is at most $T$. For some $i\in E$ we have that $\tau_i(\alpha)\geq 0$ and $(1-\alpha)v_ip_i\geq U$. Then there exists a linear contract $t_{\alpha'}$ such that $U_P(t_{\alpha'})\geq U/(2n)^T$.
\end{lemma}

\begin{proof}
Let us prove the statement by contradiction.
Let us define $U_0:=U/2$, $\phi_0:=\alpha+\frac{1-\alpha}{2}$, $i_0:=i$. Note that we have the following properties for $q=0$
\[
    \tau_{i_q}(\phi_q)\geq U_q
\]
and 
\[
    (1-\phi_q) v_{i_q} p_{i_q}\geq U_q\,,
\]
because $\tau_i(\alpha)\geq 0$ and $(1-\alpha)v_ip_i\geq U$. 

For the sake of contradiction, let us consider $q=0$ and assume that $U_P(\phi_q)\leq U_q/n$. Then with probability at least $1/n$ an element $j$ was probed and was accepted before the element  $i_q$ was probed, and moreover $v_j<v_{i_q}$. Since the element $j$ was probed before the element $i_q$ we have $\tau_j(\phi_q)\geq \tau_{i_q}(\phi_q)\geq U_q$. Let $\mu$ be the value such that $\tau_j(\mu)=0$. For $\tau_j(\phi_q)\geq U_q$ to hold, we need to have $\tau_j(1)\geq U_q$ and so we need to have $(1-\mu)v_j\geq U_q$. Let us define $U_{q+1}:=\frac{1}{2n} U_t$, $\phi_{q+1}:=\mu+\frac{1-\mu}{2}$, $i_{q+1}:=j$. Due to the definition of $\phi_{q+1}$, $\mu$ and due to $(1-\mu)v_j\geq U_q$, we have $\tau_{i_{q+1}}(\phi_{q+1})\geq U_q/2$ and so $\tau_{i_{q+1}}(\phi_{q+1})\geq U_{q+1}$. Due to the definition of $\phi_{q+1}$, $\mu$ and due to $p_j\geq 1/n$, we have $(1-\phi_{q+1})v_{i_{q+1}}p_{i_{q+1}}\geq \frac{1}{2n} U_q=U_{q+1}$.

Repeating this argument iteratively, we get a sequence of elements $i_0$, $i_1$, \ldots, $i_T$. However, the values $v_{i_q}$ are nonzeroes and are strictly decreasing with $q$. Since there are at most $T$ possible values we get a contradiction.

\qed\end{proof}

\begin{lemma}\label{lem:undersampling}
    Let us be given a contract $t_\alpha$ and $\mu>0$. Then with probability at least $1-\frac{1}{80 m^2n^2}$ we have
    \[
    (1-\alpha)\sum_{\substack{i\in E,\\k=1,\ldots,m }}v_{i,k} p_{i,k} \rho_{i,k}\leq 80 n^3m^3 U_P(t_\alpha)\,,
    \]
    where $\rho_{i,k}$ are the random variables computed by Algorithm~\ref{ALG:FRUGAL_SAMPLE}.
\end{lemma}
\begin{proof}
    
For each $i\in I$ and $k=1,\ldots,m$, by Markov's inequality we have 
\[
\Pr[\rho_{i,k}\geq 80 m^3n^3 r_
{i,k}]\leq \frac{1}{80 m^3n^3}\,.
\]
Thus with probability at least $1-\frac{1}{80 m^2n^2}$ for each $i\in I$ and $k=1,\ldots,m$ we have $\rho_{i,k}\leq 80 m^3n^3 r_
{i,k}$. So with probability at least $1-\frac{1}{80 m^2n^2}$ we have
\[U_P(t_\alpha)=(1-\alpha)\sum_{\substack{i\in E,\\k=1,\ldots,m }}v_{i,k} p_{i,k} r_{i,k}\geq \frac{1}{80m^3n^3} (1-\alpha)\sum_{\substack{i\in E,\\k=1,\ldots,m }}v_{i,k} p_{i,k} \rho_{i,k}\,.\]
\qed\end{proof}

\begin{proof}[Proof of Theorem~\ref{thm:FPRAS_special_cases} Case~\eqref{thm:FPRAS_special_cases support two}] Let us define $T$ to be the cardinality of the set \[\{v_{i,k}\,:\, i\in E,\, k=1,\ldots, m\}\,.
    \] 
Let us select $\mu$ such that 
\[
\mu\leq \frac{\varepsilon}{(2n)^T}\frac{1}{160 m^4n^4}\frac{1}{2\cdot 80n^3m^3}\,.
\] 
Let us call a critical value $\alpha$ ``bad" if  there exists $i\in E$ with $\tau_i(\alpha)\geq 0$ and $k=1,\ldots, m$ such that
\begin{equation}\label{eq:constant_outcomes_bound_contribution}
    (1-\alpha)v_ip_i\geq \frac{1}{160m^4n^4}U_P(t_\alpha)\frac{\varepsilon}{\mu}\,.
\end{equation}
Otherwise, we call a critical value ``good".
For each ``bad" critical value $\alpha$,   there exists a critical value $\alpha'$, $\alpha'\in [0,1]$ such that 
\[
U_P(t_{\alpha'})\geq \frac{1}{(2n)^T}\frac{1}{160m^4n^4}U_P(t_\alpha)\frac{\varepsilon}{\mu}\geq 2\cdot 80m^3n^3 U_P(t_\alpha)
\]
by Lemma~\ref{lem:change_critical_value}. 

Thus, no ``bad" critical value can correspond to an optimal linear contract, moreover the principal's utility achieved by a ``bad" critical value is at most $1/(2\cdot 80m^3n^3)$ of the optimal principal's utility. In the rest of the proof, we show that for ``good" critical values Algorithm~\ref{ALG:FRUGAL_SAMPLE} allows to estimate the principal's utility accurately and efficiently. While for ``bad" critical values  Algorithm~\ref{ALG:FRUGAL_SAMPLE} provides an estimate of the principal's utility that is lower than the optimal principal's utility.

 In total, we have at most $10n^2m^2+1$ critical values by Lemma~\ref{lem:critical_values} that can be computed efficiently. 
Thus, if for some critical value $\alpha$ the inequality~\eqref{eq:constant_outcomes_bound_contribution} holds for some $i\in E$ with $\tau_i\geq 0$ and $k=1,\ldots,m$ then $\alpha$ does not correspond to the optimal principal's utility. Moreover, the maximum principal's utility among critical values is at least $2\cdot 80m^3n^3 U_P(t_\alpha)$. 

Also by the union bound and by Lemma~\ref{lem:undersampling}, with probability at least $5/6$ for all ``bad" critical values $\alpha$,
we have
\[
    (1-\alpha)\sum_{\substack{i\in E,\\k=1,\ldots,m }}v_{i,k} p_{i,k} \rho_{i,k}\leq 80 n^3m^3 U_P(t_\alpha)\,.
\]
By union bound and by  Lemma~\ref{lem:ignoring}, with probability at least $5/6$, for all ``good" critical values $\alpha$,   we have
 \[
    (1-\varepsilon) U_P(t_\alpha)\leq(1-\alpha)\sum_{\substack{i\in E,\\k=1,\ldots,m }}v_{i,k} p_{i,k} \rho_{i,k}\leq (1+\varepsilon) U_P(t_\alpha)\,.
    \]

Hence, with probability at least $2/3$ Algorithm~\ref{ALG:FRUGAL_SAMPLE} allows to find an approximately optimal linear contract.
    
\qed\end{proof}

\begin{remark}\label{rem:suport_three reduction}
    Consider OLCPM instances where the number of values in $\{v_{i,k}\,:\, i\in E,\, k=1,\ldots,m\}$ is constant, in particular is at most five. Moreover, for each $i\in E$ the number of nonzero values $\{v_{i,k}\,:\, k=1,\ldots,m\}$ is at most two. Then such OLCPM instances allow reductions from \rpm as outlined in Section~\ref{sec:RPM-to-OLCPM}. 
\end{remark}

    Indeed, for this, one can follow the same process as in Section~\ref{sec:RPM-to-OLCPM}. All results in Section~\ref{sec:RPM-to-OLCPM} follow the same way except for Lemma~\ref{lem:RPM_reduction}. In the proof of Lemma~\ref{lem:RPM_reduction} for $i=2,\ldots,n-1$ we need to use different distributions $\mathcal F_i$. In particular,  for $i=2,\ldots,n-1$ the distribution $\mathcal F_i$ corresponds to taking the value $0$ with probability $1-p_i$, taking the value $1/\varepsilon$ with probability $q_{i,1}$ and taking the value $1/\varepsilon^{n-2}$ with probability $q_{i,2}$, where $q_{i,1}+q_{i,2}=p_i$ and $q_{i,1}/\varepsilon+q_{i,2}/\varepsilon^{n-2}= p_i/\varepsilon^{i-1}$. Changing these distributions in the proof of Lemma~\ref{lem:RPM_reduction} does not change the functions $\tau_i(\cdot)$ on the interval $[0,1]$ and does not change the critical values of interest and subsequent estimates.

\section{General Sequential Contracts}

In this paper, we study linear sequential contracts on matroids. We assume that the access to matroids is granted through an independence oracle. Let us briefly highlight challenges one encounters when instead considering general sequential contracts on matroids.

First of all, the number of possible outcomes the principal could see can be exponential in the size of the input. This makes it unclear how to specify a contract to be offered to an agent. Indeed, a general contract would require the principal to specify the amount to give to the agent, and this amount may be unique for each outcome. So it is unclear how the principal can accomplish even the task of sharing the contract with the agent succinctly.

Second, even in the cases when an optimal general sequential contract is trivial, we cannot efficiently compute the principal's expected utility when we are restricted to use only calls to the independence oracle. Let us assume that we have an even $n$ and let us identify elements in $E$ with the numbers $1,\ldots,n$. Let us consider the following values, probabilities and costs (here $k=2$): $v_{i,1}=2$, $p_{i,1}=1/2$, $v_{i,2}=0$, $p_{i,2}=1/2$ for $i=1,\ldots,n-1$; and $v_{n,1}=1$, $p_{n,1}=1/2$, $v_{n,2}=0$, $p_{n,2}=1/2$; and all costs $c_i$, $i=1,\ldots, n$ are zeroes. Since the costs of actions are zeroes and the agent is expected to break ties to the benefit of the principal, the optimal contract gives the agent zeroes for each outcome. For simplicity, let $\kappa$ be the principal's utility under this contract on the uniform matroid of rank $n/2$. 

Let us now consider a set $B$, $B\subseteq E$, $|B|=n/2$ (that is unknown to us) and a matroid $\mathcal M=(E,\mathcal I)$, where the rank of $\mathcal M$ is $n/2$ and $B$ is the only set of cardinality $n/2$ that is not independent in $\mathcal M$. It is straightforward to verify that  the principal's utility equals $\kappa-1/2^n$ if the element $n$ is in $B$; and equals $\kappa-2/2^n$ if the element $n$ is not in $B$. Indeed, this is due to the fact that reward received by the principal remains the same for the uniform matroid of rank $n/2$ and for matroid $\mathcal M$ with the only exception of the event when the set $\{i\in E\,:\, X_i\neq 0\}$
equals the set $B$. The above event happens with probability $1/2^n$; and under this event if $n$ is in $B$ the principal's utility is decreased by $1$ and if $n$ is in $B$ the principal's utility is decreased by $2$ with respect to the uniform matroid.

Thus, if we are able to compute the principal's utility under the optimal contract both in the uniform matroid and in the matroid $\mathcal M$, we would be able to determine whether a given element $n$ lies in the set $B$. Reordering the elements of $E$, we then can determine the set $B$. However, it is impossible to efficiently determine the set $B$ in the matroid $\mathcal M$ using only independence oracle calls.

\printbibliography 
\end{document}